\documentclass[runningheads]{llncs}
\usepackage{todonotes}
\usepackage{setspace}
\usepackage[ruled, vlined, linesnumbered]{algorithm2e}
\usepackage[scriptsize]{subfigure} 
\usepackage{array} 
\usepackage{multirow}
\usepackage{pict2e}
\usepackage{color}
\usepackage{balance} 
\usepackage{enumitem} 
\usepackage{float}
\usepackage[export]{adjustbox}
\usepackage{amsmath}
\usepackage{graphicx}
\usepackage{comment}

\makeatletter 
\newcommand{\raisemath}[1]{\mathpalette{\raisem@th{#1}}}
\newcommand{\raisem@th}[3]{\raisebox{#1}{$#2#3$}}
\makeatother

\begin{document}
\title{Connected-Dense-Connected Subgraphs in Triple Networks}
%
%
\author{Dhara Shah\inst{1} \and Yubao Wu\inst{1} \and
Sushil Prasad\inst{1} \and Danial Aghajarian\inst{2}}
\authorrunning{Shah et al.}
%
\institute{Department of Computer Science, Georgia State University, Atlanta 30303, USA
\\
\email{\{dshah8, ywu28, sprasad\}@gsu.edu} \and \email{daghajarian@cs.gsu.edu}}
\maketitle              
\begin{abstract}
Finding meaningful communities - subnetworks of interest within a large scale network - is a problem with a variety of applications. Most existing work towards community detection focuses on a single network. However, many real-life applications naturally yield what we refer to as Triple Networks. Triple Networks are comprised of two networks, and the network of bipartite connections between their nodes. In this paper, we formulate and investigate the problem of finding Connected-Dense-Connected subgraph (CDC), a subnetwork which has the largest density in the bipartite network and whose sets of end points within each network induce connected subnetworks. These patterns represent communities based on the bipartite association between the networks. To our knowledge, such patterns cannot be detected by existing algorithms for a single network or heterogeneous networks. We show that finding CDC subgraphs is NP-hard and develop novel heuristics to obtain feasible solutions, the fastest of which is O(nlogn+m) with n nodes and m edges. We also study different variations of the CDC subgraphs. We perform experiments on a variety of real and synthetic Triple Networks to evaluate the effectiveness and efficiency of the developed methods. Employing these heuristics, we demonstrate how to identify communities of similar opinions and research interests, and factors influencing communities.
\end{abstract}


\section{Introduction}
Community detection is a key primitive with a wide range of applications in real world \cite{fortunato2010community}. Most existing work focuses on finding communities within a single network. In many real-life applications, we can often observe Triple Networks consisting of two networks and a third bipartite network representing the interaction between them. For example, in Twitter, users form a follower network, hashtags form a co-occurrence network, and the user-hashtag interactions form a bipartite network. The user-hashtag interactions represent a user's posts or tweets containing a hashtag. Figure \ref{fig:twitter} shows a real Twitter Triple Network. The nodes on the left part represent users and those on the right represent hashtags. The edges among the nodes on the left represent a user following other user. The edges among the nodes on the right represent two hashtags appearing in the same tweet. The edges in between represent a user interacting with tweets containing a hashtag. This Triple Network model can ideally represent many real world applications such as taxi pick-up-drop-off networks, Flixster user-movie networks, and author-paper citation networks.\par
\begin{figure}[!t]
\centering
\includegraphics[width=0.45\columnwidth]{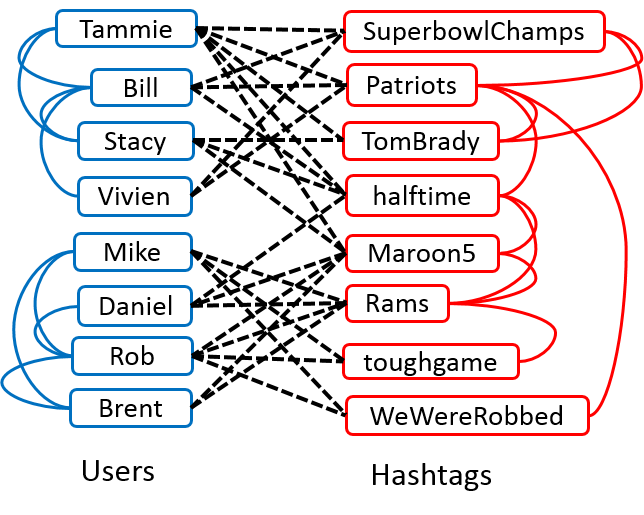}
\vspace{-10pt}
\caption{Twitter Triple Network}\label{fig:twitter}
\vspace{-2pt}
\end{figure}

In general, maximizing the density of connections in the Connected-Dense-Connected (CDC) subgraph of a  triple network is an unsupervised method for approximating the communities affiliated with the attributes.
In the twitter example the density is the number of connections between users and tweets and reflects the degree to which the users are engaged with those tweets.
Therefore finding the CDC subgraph is likely to be a useful approach to understanding social and other networks.
Given a Triple Network consisting of two graphs $G_a(V_a, E_a)$ and $G_b(V_b, E_b)$ and a bipartite graph $G_c(V_a, V_b, E_c)$, the CDC consists of two subsets of nodes $S \subset V_a$ and $T \subset V_b$ such that the induced subgraphs $G_a[S]$ and $G_b[T]$ are both connected and the density of $G_c[S, T]$ is maximized.\par
In the Twitter Triple Network in Figure \ref{fig:twitter}, we observe two CDC subgraphs: the one at the top with $S_1$ = \{Tammie, Bill, Stacy, Vivien\} and $T_1$ = \{Patriots, TomBrady, SuperbowlChamps, halftime, Maroon5\}, and the one at the bottom with $S_2$ = \{Mike, Daniel, Rob, Brent\} and $T_2$ = \{Rams, toughgame, Maroon5\}. In either of the two CDCs, the left and right networks are connected and the middle one is dense. These CDCs are meaningful. The CDC at the top shows that Patriots' fans are praising Tom Brady and are happy to be champions again. The CDC at the bottom shows that LA Rams' fans are disappointed to loose the game.\par
Our problem is different from finding co-dense subgraphs \cite{kelley2005systematic,pei2005mining} or coherent dense subgraphs \cite{hu2005mining,li2012pattern}, whose goal is to find the dense subgraphs preserved across multiple networks with the same types of nodes and edges. In our problem, the left and right networks contain different types of nodes and the edges in the three networks represent different meanings. Our problem is also different than the densest connected subgraphs in dual networks \cite{wu2015finding}. Dual networks consist of one set of nodes and two sets of edges. Triple Networks consist of two sets of nodes and three sets of edges. Triple Networks can degenerate to dual networks when the two sets of nodes are identical and the bipartite links connect each node to its replica. \par

%
%
%
\section{Background and related work}

The problem of finding a densest subgraph of a graph has been well studied by data mining community. At the core, this problem asks for finding subgraphs with the highest average degree. This problem has been solved in polynomial time using max-flow min-cut approach \cite{goldberg1984finding}. Inspired by this approach, the problem of finding densest subgraph in a directed graph has also been solved in polynomial time \cite{khuller2009finding}. The prohibitive cost of these polynomial time algorithms has been addressed with 2-approximation algorithm \cite{charikar2000greedy}. However, variations of densest subgraph problems, such as discovery of densest subgraph with $k$ nodes, have been shown to be NP-hard \cite{bhaskara2010detecting}. On the other hand, the problem of finding densest subgraph with pre-selected seed nodes is solvable in polynomial time \cite{saha2010dense}.\par
The solutions above are designed for homogeneous information network structure where the nodes and edges have just one type. Heterogeneous information networks \cite{sun2009ranking} -- the networks with multiple node and edge types -- have been a new development in the field of data mining. Heterogeneous network structure provides a model for graph infusion with rich semantics. The Triple Networks introduced in this paper are a type of heterogeneous network with node types $V_a$ and $V_b$, and edge types $E_a,E_b$ and $E_c$. Our work can be categorized as unsupervised clustering in heterogeneous network. Parallel to our work, Boden et al. discuss a density based clustering approach of k-partite graphs in heterogeneous information structure \cite{boden2014density}. In this work, two types of nodes $V_a$  and $V_b$ are considered. With node type specific hyper-parameters and the bipartite connections $E_c$, the connections $E_a$ and $E_b$ are inferred. This method of clustering is different from our work where $E_a$ and $E_b$ are part of the network, and the definition of density is hyper-parameter free. Boden et al. detect communities by subspace clustering on nodes' projection to attribute space. In contrast, our work of finding CDC subgraphs cannot be inferred as a subspace clustering technique. Though both works produce iterative refinement algorithms, the former concentrates on improving inference of $E_a$ and $E_b$ iteratively.\par 
The closest network schema to our work is dual networks \cite{wu2015finding}, discovered by Wu et al. A dual network is comprised of two networks having the same set of nodes but different types of edges. These two networks are inferred as physical and conceptual networks. Wu et al. provide 2-approximation algorithms for NP-hard problem of finding subgraphs that are densest in conceptual network, and are connected in physical network. Though the network architecture and subgraph patterns are different, our work is inspired by the pruning methods and variants proposed in this work.
To the best of our knowledge there is no comparable algorithm for finding a densest CDC subgraph.
\section{Triple network, CDC subgraphs and variants}
In this section we define Triple Network, CDC subgraph and its variants. We prove that finding CDC subgraph and variants from a Triple Network is NP-hard.

\begin{figure}[H]\vspace{-10pt}
\centering
\subfigcapskip = -2pt
\subfigure[\scriptsize An example of a toy Triple Network 
]{
\includegraphics[width=0.21\columnwidth]{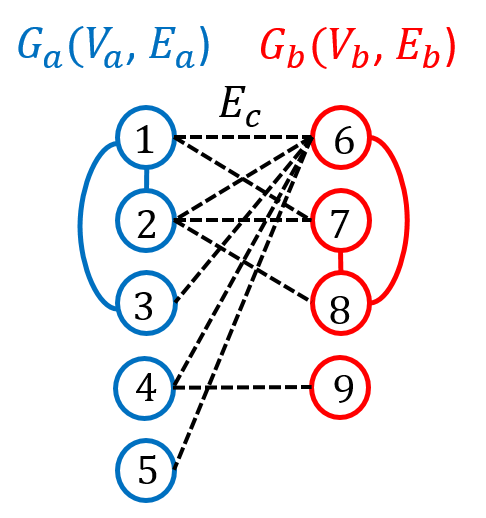}
\label{fig:toytriplenetwork}
}\hspace{6pt}
\subfigure[\scriptsize CDC subgraph of the toy Triple Network]{
\includegraphics[width=0.15\columnwidth]{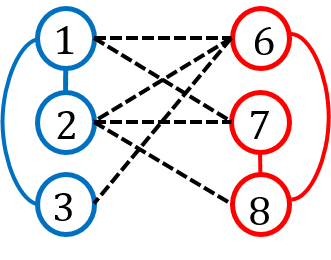}
\label{fig:toycdc}
}\hspace{6pt}
\subfigure[\scriptsize OCD subgraph of the toy Triple Network]{
\includegraphics[width=0.15\columnwidth]{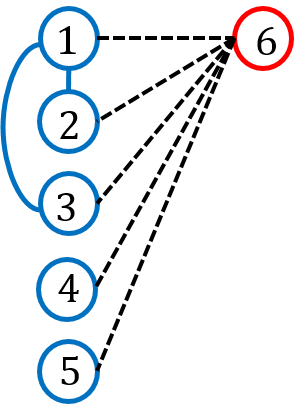}
\label{fig:toyocd}
}\hspace{6pt}
\vspace{-10pt}
\caption{Toy Triple Network and its CDC and OCD subgraphs}
\vspace{-6pt}
\end{figure}

\begin{definition}[Triple network] 
Let $G_a(V_a,E_a)$ and $G_b(V_b,E_b)$ represent graphs of two networks. Let $G_c(V_a,V_b,E_c)$ represent the bipartite graph between $G_a$ and $G_b$. $G(V_a,V_b,E_a,E_b,E_c)$ is the Triple Network generated by $G_a,G_b$ and $G_c$.
\end{definition}

We abbreviate a Triple Network as $G$. An example of Triple Network is illustrated in figure \ref{fig:toytriplenetwork}.\\

The subgraphs induced by $S_a\subset V_a$ and $S_b \subset V_b$ in networks $G_a,G_b$ and $G_c$ are denoted by $G_a[S_a]$, $G_b[S_b]$ and $G_c[S_a,S_b]$. For brevity, we denote this sub Triple Network, a set of three subgraphs, as $G[S_a,S_b]$.

\begin{definition}[Density of a Triple Network] 
Given a Triple Network $G[S_a,S_b]$, its density is defined as $
\rho(S_a,S_b)=\frac{|E_c(S_a,S_b)|}{\sqrt[]{|S_a||S_b|}}$, where $|E_c[S_a,S_b]|$ is the number of bipartite edges in subgraph $G_c[S_a,S_b]$, $|S_a|$ is the number of nodes in $G_a[S_a]$ and $|S_b|$ is the number of nodes in $G_b[S_b]$.
\end{definition}

For example, the density of sub Triple Network in figure \ref{fig:toycdc} with $S_a=\{1,2,3\}$ and $S_b = \{6,7,8\}]$ is $\rho(S_a,S_b)=\frac{|E_c(S_a,S_b)|}{\sqrt[]{|S_a||S_b|}} = \frac{6}{\sqrt[]{3*3}} = 2$.

By definition of density, only the bipartite edges of a Triple Network contribute to the density. Hence, the density of a Triple Network $G$ is same as the density of its bipartite subgraph $G_c$. 

\subsection{ Connected-Dense-Connected (CDC) subgraphs}
\begin{definition}[CDC subgraph]
Given Triple Network $G(V_a,V_b,E_a,E_b,E_c)$, a CDC subgraph is a sub Triple Network $G[S_a,S_b]$ such that 

\begin{enumerate}
\item $G_a[S_a]$ and $G_b[S_b]$ are connected subgraphs, and
\item the density $\rho(S_a,S_b)$ is maximized.
\end{enumerate}
\end{definition}

For example, the density of each CDC subgraph in figure \ref{fig:toycdc} is 2, higher than density of any other sub Triple Network of the Triple Network \ref{fig:toytriplenetwork} that is connected in $G_a$ and $G_b$. A Triple Network can have multiple CDC subgraphs. 

\begin{theorem}\label{cdcnphard}
Finding a CDC subgraph in a Triple Network is NP-Hard.
\end{theorem}
\begin{proof}
We prove that finding a CDC subgraph is a reduction of set-cover problem. Please refer to \cite{mycdc} for details. 
\end{proof}
\subsection{Variants of CDC subgraph}
CDC subgraphs stipulate connectedness of $G_a(S_a)$ and $G_b(S_b)$. Alleviating this connectivity constraint, we define OCD subgraphs for which exactly one of $G_a(S_a)$ or $G_b(S_b)$ is connected.
\begin{definition}[OCD subgraph]
Given a Triple Network $G(V_a,V_b,E_a,E_b,E_c)$ a OCD subgraph is a sub Triple Network $G[S_a,S_b]$ such that
\begin{enumerate}
\item Exactly one of $G_a[S_a]$ or $G_b[S_b]$ is connected, and
\item The density $\rho(S_a,S_b)$ is maximized.
\end{enumerate}
\end{definition}

For example, the sub Triple Network $G[\{1,2,3,4,5\},\{6\}]$ with the highest density 2.23 in figure \ref{fig:toyocd} is an OCD subgraph as $G_a[\{5\}]$ is connected. A Triple Network can have multiple OCD subgraphs.\\\\
\textbf{Adding constraints to CDC and OCD subgraphs}  We observe that CDC patterns are meaningful around pre-selected nodes in $G_a(S_a)$ or $G_b(S_b)$. We identify these pre-selected nodes as seeds. We introduce CDC and OCD subgraphs with seed constraints, where $G_a(S_a)$ or $G_b(S_b)$ should maintain their connectivity constraints while containing the seeds. 
\begin{definition}(CDC\_seeds).Given a Triple Network $G(V_a,V_b,E_a,E_b,E_c)$ and sets of seed nodes $V_1\subset V_a$ and $V_2 \subset V_b$, the CDC\_seeds subgraph consists of sets of nodes  $S_a,S_b$ such that $V_1\subset S_a $, $V_2\subset  S_b$, $G_a[S_a]$ and $G_b[S_b]$ are connected and density of $G[S_a,S_b]$ is maximized. 
\end{definition}
\begin{definition}(OCD\_seed). Given a Triple Network $G(V_a,V_b,E_a,E_b,E_c)$ and a set of node $V$ with $V\subset S$, the OCD\_seed consists of sets of nodes such that either $G_a[S]$ or $G_b[S]$ is connected and the density of $G[S,E_c[S]]$ is maximized.
\end{definition}
Finding OCD, CDC\_seeds and OCD\_seed subgraphs in a Triple Network is NP-hard. Similar set-cover arguments as in Theorem \ref{cdcnphard} could  be used to prove it. Please refer to \cite{mycdc} for details.

\section{Heuristic algorithms}
Finding CDC subgraphs is NP-hard. Hence in this section, we propose heuristic algorithms for finding feasible solutions. We propose algorithms with following two approaches. 

In the first approach, we first obtain the densest bipartite subgraph $G_c[S_a,S_b]$. We then find the connected components of $G_a[S_a]$ and $G_b[S_b]$ using BFS. Thus we obtain connected sub Triple Networks with bipartite edges in $G_c[S_a,S_b]$. We choose the highest density results as feasible CDC subgraphs. Since the time complexity of obtaining densest bipartite subgraph is higher than that of BFS, algorithms in sections \ref{MDS} and \ref{GNDR} focus on improving the complexity of finding the densest bipartite subgraphs.

In the second approach, we obtain local CDC subgraphs with given seed nodes from $V_a$ and $V_b$ by adding highest bipartite degree nodes while maintaining the connectedness in $G_a$ and $G_b$. This Local Search algorithm is presented in section \ref{LS}. 

We observe that real-world Triple Networks are sparse in $E_c$. We also observe that a connected densest subgraph exists for a bipartite graph\cite{mydensestbp}. Exploiting these virtues, we divide the bipartite graph $G_c[V_a,V_b]$ in to smaller connected bipartite subgraphs and apply the densest subgraph algorithms only for larger subgraphs. This optimization significantly reduces the running-times of our algorithms. 

\subsection{Maxflow Densest Subgraph (MDS)}\label{MDS}
MDS algorithm, formalized as Algorithm \ref{maxflow}, finds a densest bipartite subgraph of a Triple Network in polynomial time using max-flow min-cut strategies discussed in \cite{khuller2009finding} and \cite{goldberg1984finding}. We provide the details of our derivation including proofs, over all approach and examples in \cite{mydensestbp}. The density difference of any two subgraphs of a bipartite graph $G_c[V_a,V_b]$ is no less than$\frac{1}{|V_a|^2|V_b|^2}$. Hence, the binary search in MDS with step size $\frac{1}{|V_a|^2|V_b|^2}$ halts in $ O(|V_a|^{3/2}|V_b|^{3/2})$ iterations. Within each iteration, the min cut is calculated in $O(|V_a|+|V_b|)^2(2(|V_a|+|V_b|) + |E_c|))$. Hence, the complexity of MDS is $O(|V_a|^{4.5}|V_b|^{4.5})$. Adding the cost of BFS for finding connected components in $G_a$ and $G_b$, the upper-bound still remains unchanged.\par
Though polynomial time, the prohibitive time complexity of MDS algorithm makes it impracticable to employ for large Triple Networks. By using MDS results on smaller bipartite graphs as a baseline, we develop heuristics in section \ref{GNDR}. 

\begin{algorithm}[!t]
\SetAlgoVlined
\SetAlFnt{\small} 
\SetAlTitleFnt{\small} 
\SetAlCapFnt{\small}
\SetAlCapNameFnt{\small}
\SetAlCapHSkip{0ex} 
\SetAlgoCaptionSeparator{} 
\SetInd{1.6em}{1.0em} 
\SetNlSty{textsf}{}{:}
\SetNlSkip{-0.7em} 
\SetKwProg{MyRepeat}{repeat}{}{}{}
\SetKwProg{MyFor}{for}{}{}{}
\SetKwProg{MyWhile}{while}{}{}{}
\SetKwProg{myproc}{Procedure:}{}
\setstretch{}
\caption{Maxflow Densest Subgraph (MDS)} \label{maxflow}
\setstretch{1}
\begin{small}
\KwIn{Triple Network $G(V_a,V_b,E_a,E_b,E_c)$,with $V_a \neq \phi,V_b\neq \phi$}
\KwOut{A densest bi-partite subgraph $G_c[S_a,S_b]$ of $G$}
\BlankLine
\hspace{1.1em} $possible\_ratios = \{\frac{i}{j} | i \in [1,\cdots|V_a|], j \in [1, \cdots |V_b|]\}$\\
\hspace{1.1em} $densest\_subgraph = \phi,maximum\_density=\rho(V_a,V_b)$\hspace{40.6pt}\\

\hspace{1.1em}\MyFor{ ratio guess $r \in possible\_ratios$ \textbf{\textup{do}} }
{	
$low\leftarrow \rho(V_a,V_b),high\leftarrow \sqrt[]{|V_a||V_b|},g=G_c[V_a,V_b]$\\
	\SetInd{0.6em}{1em} 
	\MyWhile{$high-low\geq \frac{1}{|V_a|^2|V_b|^2}$ \textbf{\textup{do}} }
    {
    $mid = \frac{high+low}{2}$\\
    construct a flow graph $G'$ as described in \cite{mydensestbp} and find the minimum s-t cut $S,T$\\
    $g'= S\setminus \{\textrm{source node }s\}$ \\
    \DontPrintSemicolon 
	\If{$g' \neq \phi$}{$g\leftarrow g'$ \\ $low = max\{mid,\rho(g)\}$}\lElse{$high = mid$}
    \BlankLine
	\If{ $maximum\_density < low$ }{$maximum\_density = low$\\
    $densest\_subgraph=g$}
	}
}
\BlankLine
\end{small}
\end{algorithm}

\subsection{Greedy Node Deletions} \label{GNDR}
In this section, we present heuristics to obtain a dense bipartite subgraph with a reduced time complexity. 

The first heuristic is to iteratively delete the nodes with the lowest bipartite degree and yield the densest subgraph obtained in the process. This algorithm of Greedy Node Deletion using degrees (GND) is formalized as Algorithm \ref{gnd}, where criterion in line \ref{gndcriterion} is node degree. 

However, degree is not the best measure of a node's impact on density. Figure \ref{fig:toytriplenetwork} illustrates that GND deletes the nodes $\{3, 4, 5\}$ iteratively. This order of deletions leads to missing the densest bipartite subgraph $[\{1,2,3,4,5\},\{6\}]$ in figure \ref{fig:toyocd}. Instead of accounting for the connections of a node, the percent of the possible connections of that node may serve as a better measure of the node's impact on density. With this intuition, we define rank of a node.
\begin{definition}[Rank] Let $G(V_a,V_b,E_a,E_b,E_c)$ be a Triple Network. For $v_a \in V_a, rank(v_a)=\frac{d(v_a)}{|V_b|}$ and for $v_b \in V_b, rank(v_b)=\frac{d(v_b)}{|V_a|}$. 
\end{definition}
Using the lowest rank as the deletion criterion, we modify Algorithm \ref{gnd} to formulate Greedy Rank Deletion (GRD) algorithm where the criterion of deletion in line \ref{gndcriterion} is rank. 

GND and GRD delete nodes sequentially. To expedite this process, we delete all the nodes satisfying the deletion criterion in bulk in each iteration instead. This idea is formulated as fast Rank Deletion (FRD) Algorithm in \ref{frd}. These bulk deletions do not lower the time complexity upper-bound, but the number of iterations decreases exponentially. The deletion criterion of FRD could be tuned by choosing different $\epsilon$ values from $(-1,1)$ with $\epsilon$ values from lower to higher resulting in less to more deletions per iteration.

By maintaining two $\{$degree:node$\}$ Fibonacci heaps and an index on the nodes, the time complexity of these greedy deletion algorithms is $O((V_a+V_b)log(V_a+V_b)+E_c)$. Adding the cost of BFS for connected components in $G_a$ and $G_b$, the total time complexity for obtaining CDC subgraphs is $O((V_a+V_b)log(V_a+V_b)+E_c+E_a+E_b)$.\\
\begin{minipage}{6cm}
\setlength{\algomargin}{0em}
\setlength{\algoheightrule}{0.4pt} 
\setlength{\algotitleheightrule}{0.4pt} 
  \begin{algorithm}[H]
\SetAlgoVlined
\SetAlFnt{\small} 
\SetAlTitleFnt{\small} 
\SetAlCapFnt{\small}
\SetAlCapNameFnt{\small}
\SetAlCapHSkip{0ex} 
\SetAlgoCaptionSeparator{} 
\SetInd{1.6em}{1.0em} 
\SetNlSty{textsf}{}{:}
\SetNlSkip{-0.7em} 
\SetKwProg{MyRepeat}{repeat}{}{}{}
\SetKwProg{MyFor}{for}{}{}{}
\SetKwProg{MyWhile}{while}{}{}{}
\setstretch{0.5}
\caption{Greedy Node Deletions} \label{gnd}
\setstretch{1}
\begin{small}
\KwIn{Triple Network $G(V_a,V_b,E_a,E_b,E_c)$,with $V_a \neq \phi,V_b\neq \phi$,\\
\hspace{1.1em} \textit{criterion} to delete nodes}
\KwOut{A densest subgraph $G_c[S_a,S_b]$ of $G$}
\BlankLine
\hspace{1.1em} $S_a=V_a,S_b=V_b$\\
\hspace{1.1em} $maximim\_density=\rho(V_a,V_b)$\hspace{10.6pt}\\
\hspace{1.1em}\While{$V_a \neq \phi$ and $V_b \neq \phi$}{
$v = $ node with minimum \textit{criterion} in $V_a\cup V_b$\\ \label{gndcriterion}
$V_a=V_a\setminus \{v\},V_b=V_b\setminus \{v\}$\\
\If{$maximum\_density<\rho(V_a,V_b)$}{
	$S_a=V_a,S_b=V_b$,\\$E_c=E_c[V_a,V_b]$
	}
    }
    \hspace{1.1em}\KwRet{$G_c[S_a,S_b]$}
\end{small}
  \end{algorithm}
\end{minipage}%
\hfill
\begin{minipage}{6cm}
  \setlength{\algomargin}{0em}
\setlength{\algoheightrule}{0.4pt} 
\setlength{\algotitleheightrule}{0.4pt} 
\begin{algorithm}[H]
\SetAlgoVlined
\SetAlFnt{\small} 
\SetAlTitleFnt{\small} 
\SetAlCapFnt{\small}
\SetAlCapNameFnt{\small}
\SetAlCapHSkip{0ex} 
\SetAlgoCaptionSeparator{} 
\SetInd{1.6em}{1.0em} 
\SetNlSty{textsf}{}{:}
\SetNlSkip{-0.7em} 
\SetKwProg{MyRepeat}{repeat}{}{}{}
\SetKwProg{MyFor}{for}{}{}{}
\SetKwProg{MyWhile}{while}{}{}{}
\setstretch{0.5}
\caption{Fast Rank Deletion (FRD)} \label{frd}
\setstretch{1}
\begin{small}
\KwIn{Triple Network $G(V_a,V_b,E_a, E_b, E_c)$,with $V_a \neq \phi,V_b\neq \phi,$\\ 
\hspace{1.3em}value of $\epsilon \in (-1,1)$}
\KwOut{A densest bi-partite subgraph $G_c[S_a,S_b]$ of $G$}
\BlankLine
\hspace{1.1em} $S_a=V_a,S_b=V_b,$\\
\hspace{1.1em}$maximim\_density=\rho(V_a,V_b)$\\
\hspace{1.1em}\While{$V_a \neq \phi$ and $V_b \neq \phi$}{
$\Bar{r}=$ average node rank in $G$\\
$\Bar{V}\! =\!\{\! v\!\in\!V_a\!\cup\!V_b\!\mid\!rank(v)\!<\!(1+\epsilon)\Bar{r} \}$\\
$V_a=V_a\setminus\Bar{V},V_b=V_b\setminus \Bar{V}$\\
\If{$maximum\_density<\rho(V_a,V_b)$}{
	$S_a=V_a,S_b=V_b$,\\
	$E_c=E_c[V_a,V_b]$
	}
    }
    \hspace{1.1em}\KwRet{$G_c[S_a,S_b]$}
\end{small}
\end{algorithm}
\end{minipage}
\begin{algorithm}[H]
\SetAlgoVlined
\SetAlFnt{\small} 
\SetAlTitleFnt{\small} 
\SetAlCapFnt{\small}
\SetAlCapNameFnt{\small}
\SetAlCapHSkip{0ex} 
\SetAlgoCaptionSeparator{} 
\SetInd{1.6em}{1.0em} 
\SetNlSty{textsf}{}{:}
\SetNlSkip{-0.7em} 
\SetKwProg{MyRepeat}{repeat}{}{}{}
\SetKwProg{MyFor}{for}{}{}{}
\SetKwProg{MyWhile}{while}{}{}{}
\SetKwRepeat{Do}{do}{\hspace{0.5em} while}%
\setstretch{}
\caption{Local Search (LS) } \label{localsearch}
\setstretch{1}
\begin{small}
\KwIn{ $G(V_a,V_b,E_a, E_b,E_c)$,with $V_a \neq \phi,V_b\neq \phi$\\\hspace{1.1em} $seedS_a$ = Set of seeds in $V_a$\\\hspace{1.1em} $seedS_b$ = Set of seeds in $V_b$\\} 
\KwOut{A sub Triple Network $G[S_a,S_b]$ of $G$}
\BlankLine
\hspace{1.1em} $S_a$ = Spanning tree of $seedS_a$ in $G_a$\\
\hspace{1.1em} $S_b$ = Spanning tree of $seedS_b$ in $G_b$\\
\hspace{1.1em} $\delta(S_a)$ = $\{v\not\in S_a|\text{ $S_a$ contains $v$'s neighbor in $G_a$}\}$,Boundary of $S_a$ in $G_a$\\
\hspace{1.1em} $\delta(S_b)$ = $\{v\not\in S_b|\text{ $S_b$ contains $v$'s neighbor in $G_b$}\}$, Boundary of $S_b$ in $G_b$\\
\hspace{1.1em} $nbhd$, the adjacency list of $V_a$ in $G_a$ and $V_b$ in $G_b$\\
\hspace{1.1em}$max\_density = \rho(G[S_a,S_b])$ \\
\hspace{1.1em}\Do{
$\rho(G_c[S_a,S_b])\geq max\_density$ and $\delta(S_a)\cup\delta(S_b) \neq \phi$
}{
$v$ = node in $\delta(S_a)\cup \delta(S_b)$ with the highest bi-partite connections to $S_a \cup S_b$\\
$S_a = S_a \cup v $ if $v \in V_a$, $S_b = S_b \cup v $ if $v \in V_b$\\
$\delta(S_a)\cup\delta(S_b) = \delta(S_a)\cup\delta(S_b)\cup nbhd(v) \setminus \{v\} $\\
$max\_density = max(max\_density,\rho(G[S_a,S_b]))$
}
\hspace{1.1em} return $G[S_a,S_b]$
\end{small}
\end{algorithm}
\subsection{Local Search}\label{LS}
In this section, we introduce Local Search (LS), a bottom-up approach for obtaining CDC subgraphs around seeds -- pre-selected nodes. Let $S_a$ and $S_b$ be the spanning trees of desired seeds in $V_a$ and $V_b$. LS, outlined as Algorithm \ref{localsearch}, iteratively includes previously un-included boundary node of $S_a \cup S_b$ with the maximum adjacency value to the set of included nodes. LS hence finds CDC subgraph by adding nodes that increase the density while maintaining connectedness of $S_a$ and $S_b$.  

As illustrated in experiments, LS yields local patterns with good semantic value. In practice, the search stops in a few iterations and hence LS is emperially the fastest algorithm yet.
\subsection{Algorithms for variants}

We obtain OCD subgraphs as bi-products of mining CDC subgraphs. For MDS and Greedy Node Deletions, the resultant sub Triple Networks maintaining exactly one connectedness with the highest density are yielded as OCD subgraphs. We instantiate LS algorithm with either $S_a$ or $S_b$ to be empty and obtain CDC\_seeds and OCD\_seed subgraphs.  

\section{Experiment results}
In this section, we evaluate the effectiveness and efficiency of the proposed methods through comprehensive experiments on real and synthetic datasets. We demonstrate the effectiveness of CDC and OCD subgraphs by illustrating novelty of the information obtained from these subgraphs on real Triple Networks. We demonstrate the efficiency of our algorithms by measuring the running times of the algorithms and the density of the resultant CDC subgraphs. The programming language employed is Python 2.7 and the experiments were conducted on Intel Core i7 3.6Gz CPU with 32G memory.

\subsection{Real Triple Networks}
We employ Triple Networks constructed from Twitter, NYC taxi data, Flixter and ArnetMiner coauthor datasets. Table \ref{tab:networksdescription} describes the statistics of these real Triple Networks. \\
\textbf{NYC Taxi data} New York City (NYC) yellow cab taxi data is a dataset \cite{nyc_data_link} where each taxi trip's pick-up and drop-off point is a geographic location in decimal degrees. We consider the trips from June 2016 to construct a Triple Network. The geographic location accuracy of this dataset is thresholded up to 5 decimal points, preserving granularity to different door-entrances. Hence $G_a$ and $G_b$ are the networks of pick-up and drop-off points. In these networks, edges connect the points within 50 meters of haversine distance. The taxi trips are represented as $E_c$. \\
\textbf{Twitter network} Twitter is a social media for micro-blogging where users can follow each other for updates. To extract meaningful user-follower relationships, we choose popular news networks, namely CNN, Huffington Post and  Fox News, and randomly extract a few thousand of their intersecting followers. We iteratively grow this network by including followers of existing nodes using Twitter's REST API. At each iteration, we threshold users by number of recent tweets and number of followers. Thus, we construct a 5-hop users-followers network $G_a$, where two users are connected if one follows the other. We collect different hashtags from these users' tweets with $E_c$ as users posting hashtags. We consider two hashtags connected if they appear in the same tweet, and thus construct hashtag co-occurance network as $G_b$.\\
\textbf{ArnetMiner Coauthor data} ArnetMiner Coauthor dataset \cite{tang2008arnetminer} is comprised of two types of networks: authors and their co-author relationships as $G_a$, and their research interests as $G_b$, with $E_c$ as relations of authors to their research interests. We consider two research interests linked if they co-occur in an other's list of research interests.\\
\textbf{Flixter data} Flixter \cite{jamali2010matrix} is a social network of users and their movie ratings. We consider the users social network as $G_a$, the users' rankings of movies as $E_c$, and movies as $V_b$. With no sufficient information, we consider $|E_b|=0$.

\vspace{-20pt}
\begin{table}
\parbox{.58\linewidth}{
\centering
\small{
\caption{Real triple-networks on NY Taxi data (TX), Twitter (TW), ArnetMiner (AM), and Flixter (FX) data}\label{tab:networksdescription}
\setlength\tabcolsep{1.5pt} 
\begin{tabular}{|r|r|r|r|r|r|}
\hline
Data & $|V_a|$&$|E_a|$&$|V_b|$&$|E_b|$&$|E_c|$\\
\hline
TX & $733896$ & $31513503$& $794085$ & $13465065$ & $2066569$  \\
\hline
TW & $61726$ &$7008491$ &$3679824$ &$2896925$&$48269139$ \\
\hline
AM & $1712433$ &$4258946$ & $3901018$ &$953490$ &$12589981$\\
\hline
FX & $786936$ & $7058819$ & $48794 $ & $0$ & $8196077$\\
\hline
\end{tabular}
}
}
\hfill
\parbox{.38\linewidth}{
\centering
\small{
\caption{Synthetic Random and R-MAT networks \label{tab:syntheticnetworksdescription}}
\begin{tabular}{|c|c|c|c|c|}
\hline
 $|V_a|$&$|E_a|$&$|V_b|$&$|E_b|$&$|E_c|$\\
\hline
 $2^{19}$ &$5\times 10^6$ & $2^{19}$ &$5\times 10^6$ &$10^7$\\
\hline
 $2^{20}$ &$10^7$ & $2^{20}$ &$10^7$ &$2 \times 10^7$\\
\hline
 $2^{21}$ &$2 \times 10^7$ & $2^{21}$ &$2 \times 10^7$ &$4 \times 10^7$\\
\hline
 $2^{22}$ &$4 \times 10^7$ & $2^{22}$ &$4 \times 10^7$ &$8 \times 10^7$\\
\hline
\end{tabular}
}
}
\vspace{-20pt}
\hspace{-20pt}
\end{table}
\subsection{Synthetic Triple Networks}
We generated random networks with synthetic $G_a, G_b$ and $G_c$ having random edges in order to evaluate efficiency of our algorithms. 
To approximate real world Triple Networks, we also generated R-MAT networks with $G_a$ and $ G_b$ having R-MAT edges \cite{chakrabarti2004r,bader2006gtgraph} and $G_c$ having random edges. We generated four different configurations for random and R-MAT networks (see Table \ref{tab:syntheticnetworksdescription}).
To the best of our knowledge, there are no algorithms to obtain CDC subgraphs. 
However, the MDS algorithm provides the densest bipartite subgraph, and hence is an upper-bound to the density of CDC. 
The high time complexity of MDS algorithm limits its applicability with real problems and thus we used synthetic benchmarks.
\begin{table}[]
    \centering
    \begin{tabular}{|r|r|r|r|r|}
        \hline
        $|V_a|=|V_b|$&$|E_a| = |E_b|$&$|E_c|$ & \begin{tabular}{@{}c@{}}Random networks:\\ MDS/GRD bipartite \end{tabular} &\begin{tabular}{@{}c@{}}RMAT networks:\\ MDS/GRD bipartite \end{tabular}\\
        \hline
         $2^{15}$ &$3.125 \times 10^5$ &$ 6.25 \times 10^5$& 0.9897 & 1.1970\\
        \hline
        $2^{16}$ &$6.25 \times 10^5$ & $ 1.25 \times 10^6$ & 0.9901 & 1.1898\\
        \hline
        $2^{17}$ &$1.25 \times 10^6$ & $ 2.5 \times 10^6$ & 0.9865 & 1.2101\\
        \hline
        $2^{18}$ &$ 2.5 \times 10^6$ & $5 \times 10^6$ & 1.0010 &  1.1985\\
        \hline
        $2^{19}$ &$5 \times 10^6$ & $10^7$ & 0.9753 &  1.2021\\
        \hline
    \end{tabular}
    \caption{Caption}
    \label{tab:my_label}
\end{table}

\subsection{Effectiveness Evaluation on Real Networks}\label{Effective}
\begin{figure}[!t]\vspace{2pt}
\centering
\subfigcapskip = -2pt
\subfigure[\scriptsize CDC subgraph yielding directional flow of human migration in 1 hour period
]{
\includegraphics[width=0.34\linewidth, height = 1.3 in]{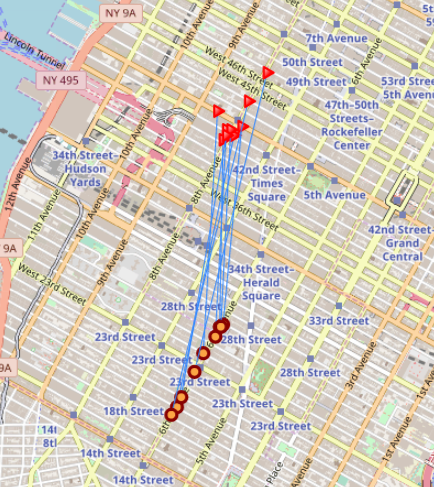}
\label{fig:nycdc}
}\hspace{2pt}
\subfigure[\scriptsize OCD subgraph yielding drop-off hot-spots on a street in 4 hours period]{
\includegraphics[width=0.34\linewidth, height = 1.3 in]{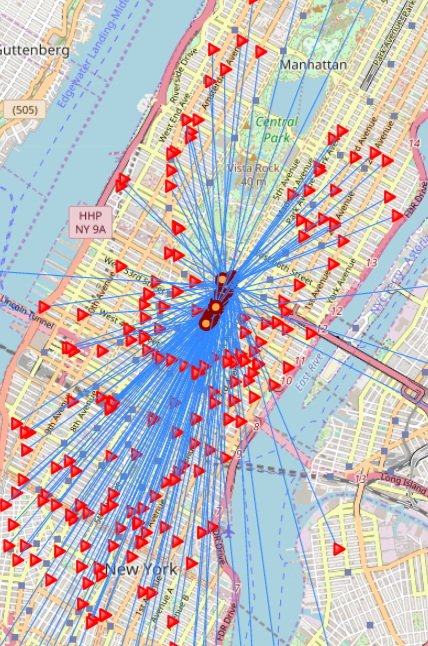}
\label{fig:nyocd}
}\hspace{2pt}
\vspace{-10pt}
\caption{CDC and OCD subgraphs from NY Taxi data. Triangles and circles represent pick-up and drop-off points respectively}
\vspace{-10pt}
\label{fig:nytaxi}
\end{figure}
\begin{figure}[!t]\vspace{2pt}
\centering
\subfigure[{\scriptsize CDC subgraph representing Patriots' fans}]{

\includegraphics[width=0.57\linewidth]{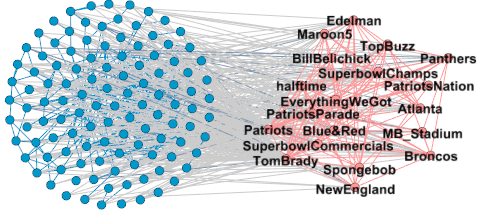}\label{fig:democrats}
}
\hfill
\subfigure[{\scriptsize CDC subgraph representing Rams' fans}]{

\includegraphics[width=0.34\linewidth]{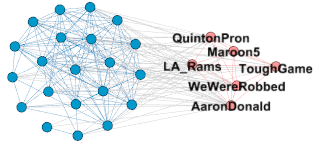}\label{fig:republicans}
}
\vspace{-12pt}
\vspace{-2pt}
\caption{CDC subgraphs from Twitter. Users-followers networks on the left and hashtag networks on the right.}
\label{fig:twitter}
\end{figure}
\begin{figure}[!t]\vspace{2pt}
\centering
\subfigure[{\scriptsize CDC\_seeds subgraph with author and research-interests seeds}]{
\includegraphics[width=0.5\linewidth]{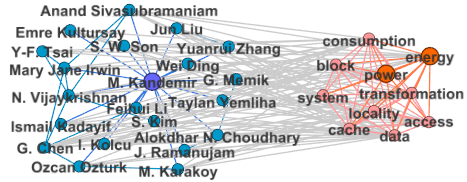}\label{fig:aminercdc}
}
\hfill
\subfigure[{\scriptsize  OCD\_seed subgraph with research-interest seeds }]{
\includegraphics[width=0.37\linewidth]{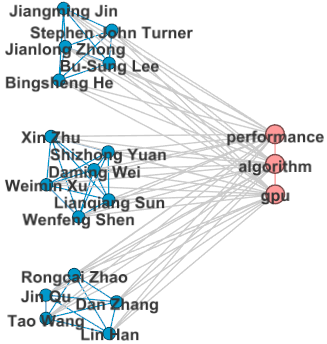}\label{fig:aminercdctopicseeds}
}
\vspace{-12pt}
\vspace{-2pt}
\caption{CDC and OCD subgraphs from ArnetMiner. Co-author networks on the left and research-interest networks on right. }
\label{fig:aminer}
\end{figure}
\begin{figure}[!t]\vspace{2pt}
\centering
\subfigcapskip = -2pt
\subfigure[\scriptsize OCD\_seed subgraph of user seeds influenced by movies]{
\includegraphics[width=0.35\linewidth]{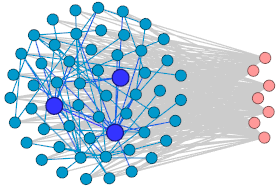}
\label{fig:flixteruserseeds}
}
\hspace{8pt}
\subfigure[\scriptsize OCD subgraph of a possible fraud]{
\includegraphics[width=0.35\linewidth]{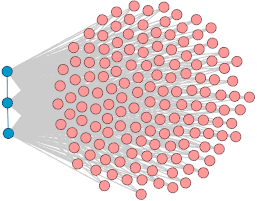}
\label{fig:flixterfraud}
}\hspace{8pt}
\vspace{-12pt}
\caption{ OCD subgraphs from Flixter. User networks on the left and movie networks on the right.}
\label{fig:flixter}
\end{figure}
We illustrate the effectiveness of CDC subgraphs and variants by emphasizing the knowledge gain from these patterns obtained from real networks. These figures demonstrate that CDC subgraphs and variants are communities detected by the strong associations to their attributes. These subgraphs identify similar opinions, research interests and factors influencing communities. They are also effective tools for hot-spot detection and fraud detection.\\
\textbf{NYC Taxi data} Figure \ref{fig:nytaxi} illustrates CDC and OCD subgraphs with pick-up and drop-off points as triangles and circles respectively.\par
Figure \ref{fig:nycdc} illustrates the CDC subgraph with pick-up locations on 6th Avenue between 18th and 27th street populated with food and shopping destinations, and drop-of locations on 8th Avenue. This CDC subgraph is generated by observing the 6:00-7:00 pm traffic on June 4, 2016. The drop-off points are clustered near 42nd street Port Authority bus terminals of city transit. This CDC subgraph gives a directional flow of human migration in a short distance during a specific time-frame. Figure \ref{fig:nyocd} illustrates OCD subgraph with pick-up seeds near 5th Avenue and Central Park South. This subgraph is generated by observing 4:00-8:00 pm traffic on June 1, 2016. The pick-up points are scattered along Manhattan and the drop-off points are clustered around Pennsylvania Station, a public transit hub. Thus, OCD subgraphs could be equivalents to hot-spot detection.\\
\textbf{Twitter Network} Figure \ref{fig:twitter} represents CDC subgraphs obtained from Twitter Network. Left and right subgraphs represent users-followers and hashtag networks. We remove usernames to protect user privacy. These figures represent twitter users and their opinions about SuperBowl contenders, Patriots and LA Rams. Hence, CDC subgraphs can identify communities with contrasting opinions.\\
\textbf{ArnetMiner coauthor data} Figure \ref{fig:aminer} depicts CDC\_seeds and OCD\_seed subgraphs from ArnetMiner Triple Network. Left and right subgraphs represent author-coauthor and research-interest networks. 

Figure \ref{fig:aminercdc} is a CDC\_seeds subgraph with randomly chosen author seed $\{$M.Kandimir$\}$ and interest seeds $\{$power,energy$\}$. This pattern yields author seed's associates working on related research topics of interest seeds. Figure \ref{fig:aminercdctopicseeds} is OCD\_seed subgraph with interest seeds chosen as $\{$algorithm, gpu, performance$\}$. This patterns yields 16 authors and their respective co-author networks with publications related to interests seeds. Thus, even with the given seeds, the CDC and OCD subgraphs are different from supervised community detection.\\
\textbf{Flixter data} Figure \ref{fig:flixter} depicts OCD subgraphs illustrating influence of movies on users. Left and right subgraphs represent the users' social networks and the movies networks, The users networks are connected. 

Figure \ref{fig:flixteruserseeds} is an OCD\_seed subgraph with users seeds, chosen at random. The right network represents movies with 5 star rankings by the users on the left. This pattern hence finds the movies influencing the friend-circle of the seed users. An OCD subgraph in figure \ref{fig:flixterfraud} depicts a suspicious ranking activity, where the 3 users on the left give a 5 star ranking to 144 movies on the right. CDC and OCD subgraphs hence illustrate the power of potential fraud detection.

\subsection{Efficiency evaluation}
We evaluate the efficiency of our heuristic algorithms by their running-time and the quality of the resulting CDC subgraphs from real and synthetic networks. \\\\
\textbf{Greedy node deletions} The running-times of  MDS, GND, GRD, FRD algorithms on real, random and R-MAT networks are depicted in Figure \ref{fig:topdownruntimes}. The x axis represents the number of nodes in $V_a \cup V_b$ and the y axis represents log scale of seconds. Each point represents running-time of the algorithm for given network. The running-time of MDS algorithm for larger networks is more than 24 hours, when we halted the algorithm computations. Running-times increase with network size, but vary a little for random and R-MAT graphs of the same size. FRD with $\epsilon = 0 $ is the fastest algorithm.\par
We discover that GRD yields the densest bipartite subgraph among all algorithms. The densities of CDC subgraphs obtained by GND, GRD and FRD from random and R-MAT networks are presented in table \ref{tab:randomdensities} and \ref{tab:rmatdensities}. For each graph, DBP represents the density of the densest bipartite graph obtained by GRD, without being connected in $G_a$ or $G_b$. The ratio, DBP/CDC densitiy, varies a little with the network size. This trend is observed across all network types and algorithms. GRD produces the best and FRD with $\epsilon = 0$ produces the least accurate results. \\\\
\begin{figure}[t]
\centering
\subfigcapskip = -2pt
\subfigure[\scriptsize Random networks]{
\includegraphics[height=1.1in]{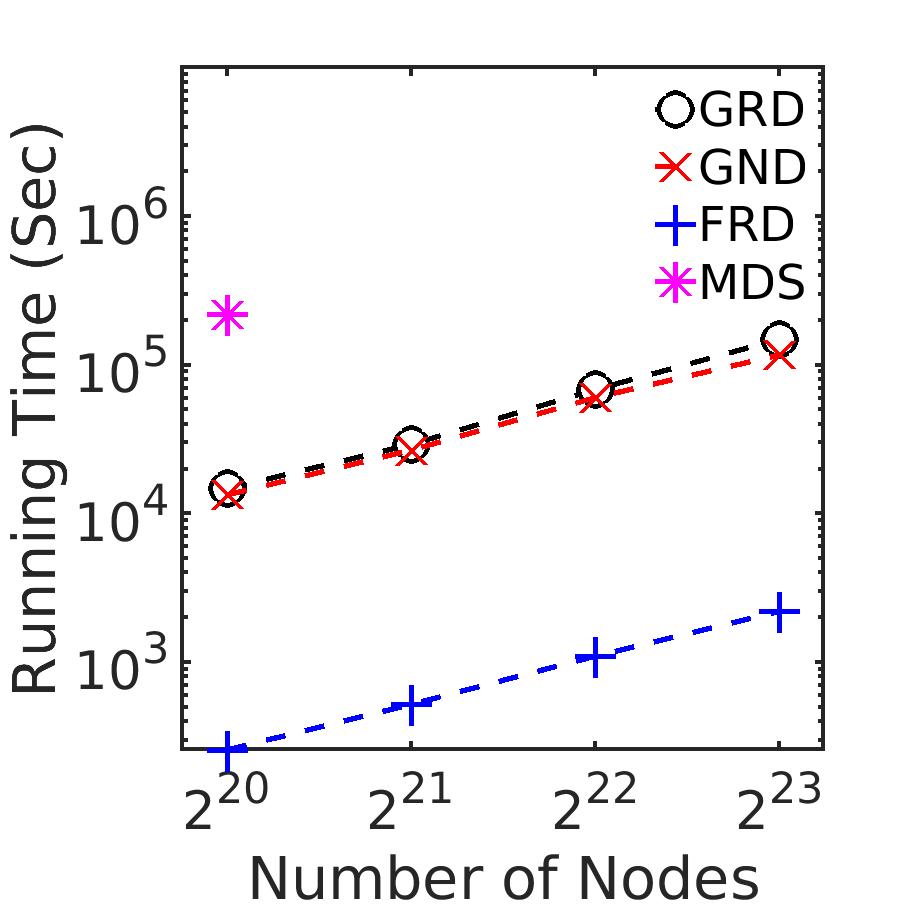}
}\hspace{4pt}
\subfigure[R-MAT networks]{
\includegraphics[height=1.1in]{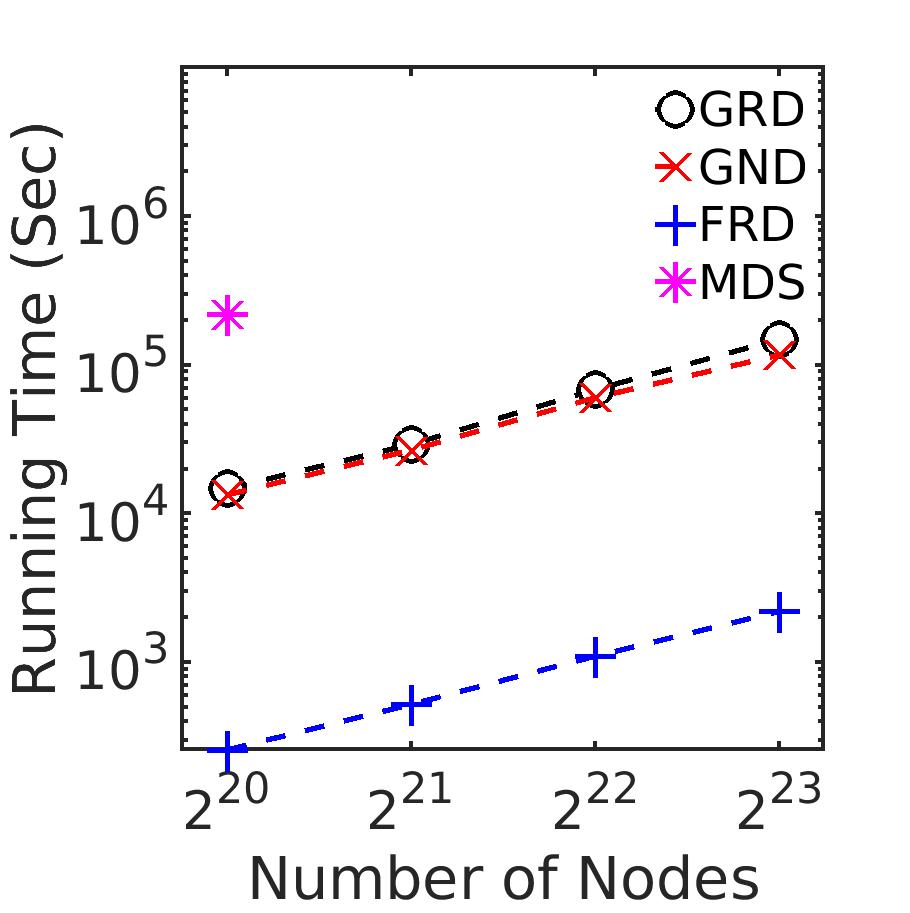}
}\hspace{4pt}
\subfigure[Real Networks]{
\includegraphics[height=1.3in]{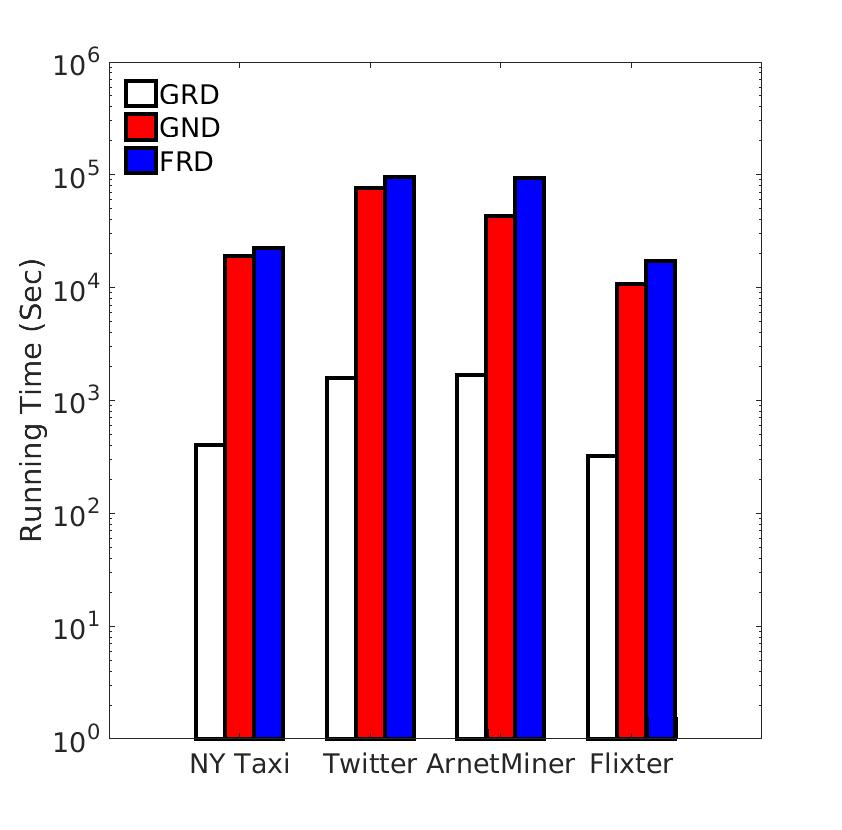}
}\label{fig:realnetrunningtimes}
\vspace{-12pt}
\caption{Running-times for MDS, GND, GRD and FRD }\label{fig:topdownruntimes}
\vspace{-2pt}
\end{figure}
\begin{table}[]
\parbox{.48\linewidth}{
\centering
\small{
\caption{CDC subgraph densities from random networks}\label{tab:randomdensities}
\setlength\tabcolsep{1.5pt} 
\begin{tabular}{|r|r|r|r|r|}\hline
 & $2^{20}$ & $2^{21}$ & $2^{22}$ & $2^{23}$\\\hline
DBP& 19.083 &	19.095 &	19.094&	19.086\\\hline
GND& 18.713 &	18.705 & 18.691	& 18.720\\\hline
GRD& 18.901 &	18.836 & 18.837	& 18.698\\\hline
FRD& 7.401 & 7.389 & 7.402 & 7.401\\\hline
\end{tabular}
}
}
\hspace{5pt}
\parbox{.48\linewidth}{
\centering
\small{
\caption{CDC subgraph densities from R-MAT networks \label{tab:rmatdensities}}
\begin{tabular}{|r|r|r|r|r|}\hline
 & $2^{20}$ & $2^{21}$ & $2^{22}$ & $2^{23}$\\\hline
 DBP& 19.071 &	19.065 &	19.073 & 19.072\\\hline
GND& 17.028	& 16.761 & 17.019 &	16.627\\\hline
GRD& 17.201	& 17.002 &	17.046 & 16.689\\\hline
FRD& 6.612 & 6.610 &	6.509 & 6.501\\\hline
\end{tabular}
}
}
\hspace{-20pt}
\end{table}
\textbf{Local Search (LS)} Given the seeds of $V_a$ and $V_b$, LS produces meaningful, locally dense CDC patterns. We evaluate the efficiency of LS algorithm by measuring its running-times with  2, 4 and 8 seeds. Figure \ref{fig:lsruntimes} presents the running-times of LS. The x axis represents the number of nodes in $V_a \cup V_b$ and the y axis represents running-times in seconds. Each point represents running-time of FRD for given network and seed configuration. The seeds are chosen randomly in the same connected components. The boundaries $\delta(S_a)$ and $\delta(S_b)$ grow larger with increase in the number of seeds. Hence the running-time of LS increases with the number of seeds.  We observe similar trends from real networks. In synthetic networks, for a given number of seeds, LS running-times vary a little across different network sizes. This is because LS halts when the density of the current CDC subgraph starts decreasing, which depends only on the local topologies of $G_a$ and $G_b$. \\\\
\begin{figure}[!t]\vspace{2pt}
\centering
\subfigcapskip = -2pt
\subfigure[\scriptsize Random networks]{
\includegraphics[height=1.1in]{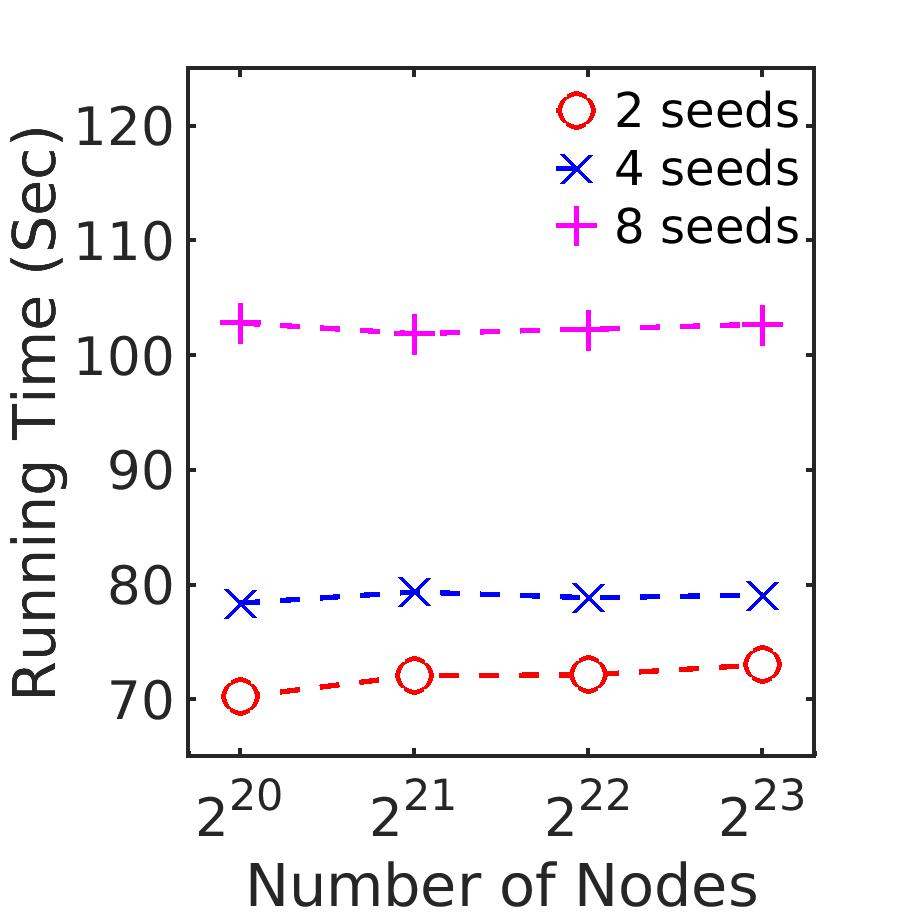}
}\hspace{6pt}
\subfigure[\scriptsize R-MAT networks]{
\includegraphics[height=1.1in]{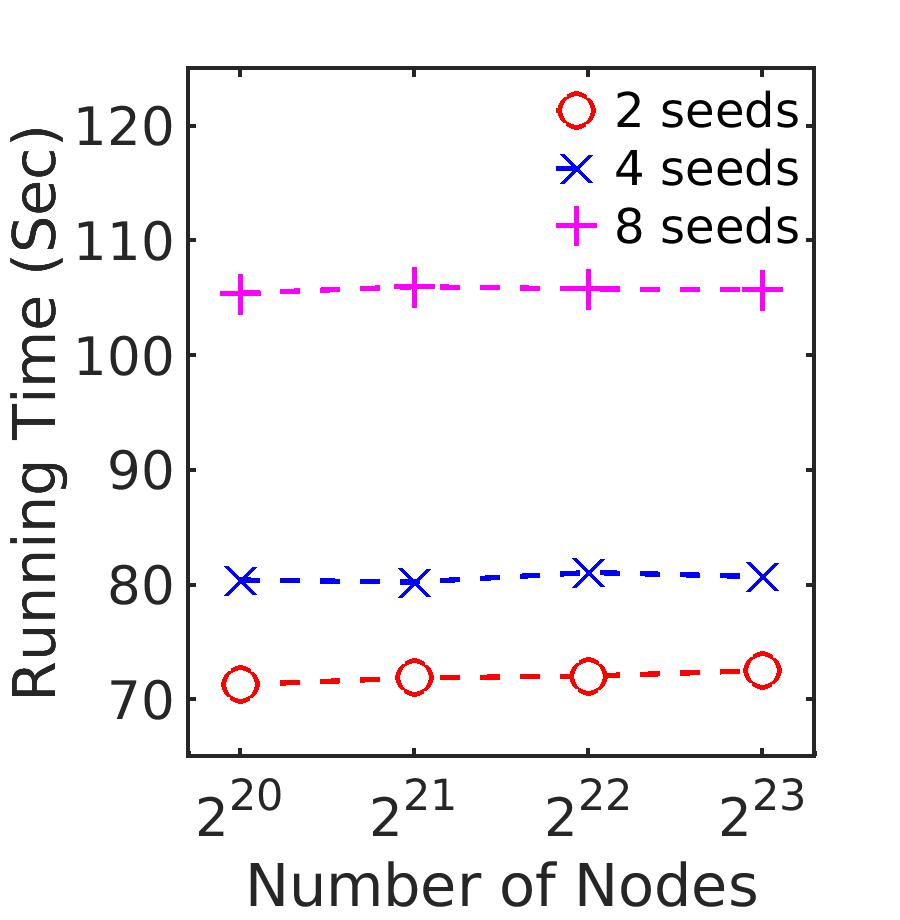}
}
\vspace{-12pt}
\caption{LS running-times with 2,4 and 8 seeds }\label{fig:lsruntimes}
\vspace{-2pt}
\end{figure}
\textbf{Fast Rank Deletion (FRD)} The purpose of FRD is to obtain feasible CDC subgraphs faster. This is achieved by deleting all the nodes with degree less than $(1+\epsilon)*\textit{average degree}$ at each pass. However, lower $\epsilon$ values result in fewer deletions per pass, defying the purpose of FRD. Higher $\epsilon$ values result in more deletions per pass, lowering the densities of the resulting CDC subgraphs. Hence the meaningful results are obtained with $\epsilon$ values in the range of interval $[-0.4,0.4]$.
\begin{figure}[!t]\vspace{2pt}
\centering
\subfigcapskip = -2pt
\subfigure[\scriptsize Random networks]{
\includegraphics[height=1.1in]{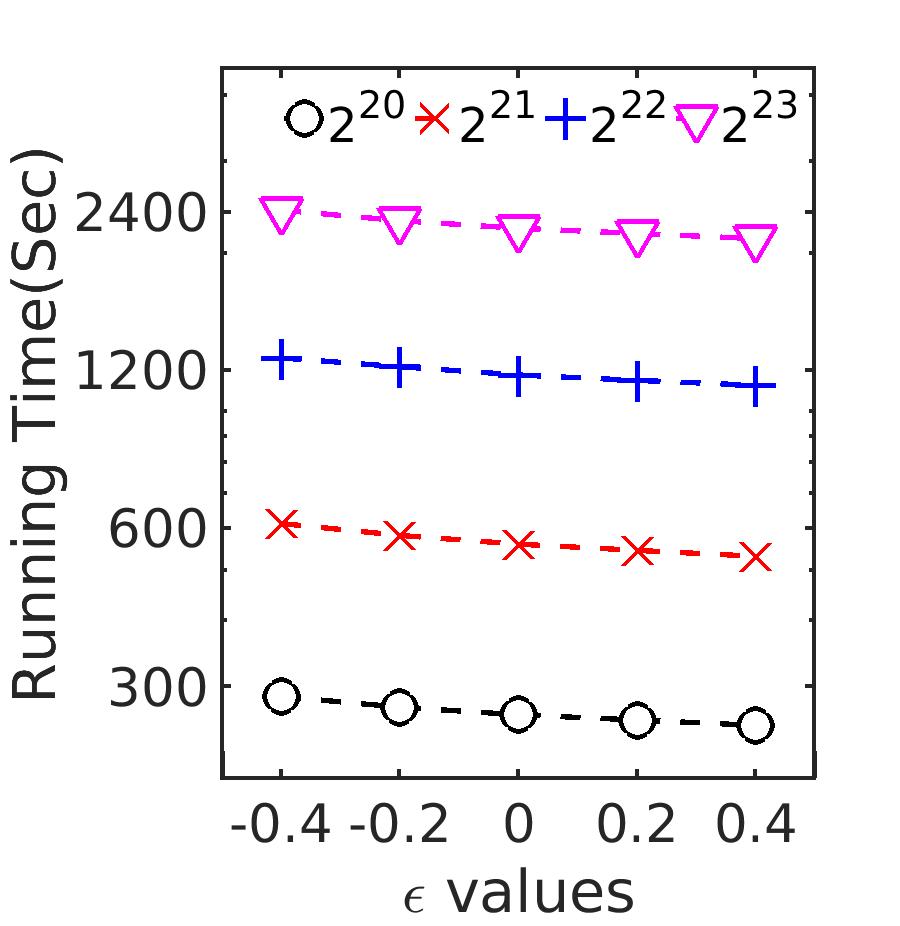}\label{fig:frdrunrandom}
\label{fig:frdruntimerandom}
}\hspace{-15pt}
\subfigure[\scriptsize R-MAT networks]{
\includegraphics[height=1.1in]{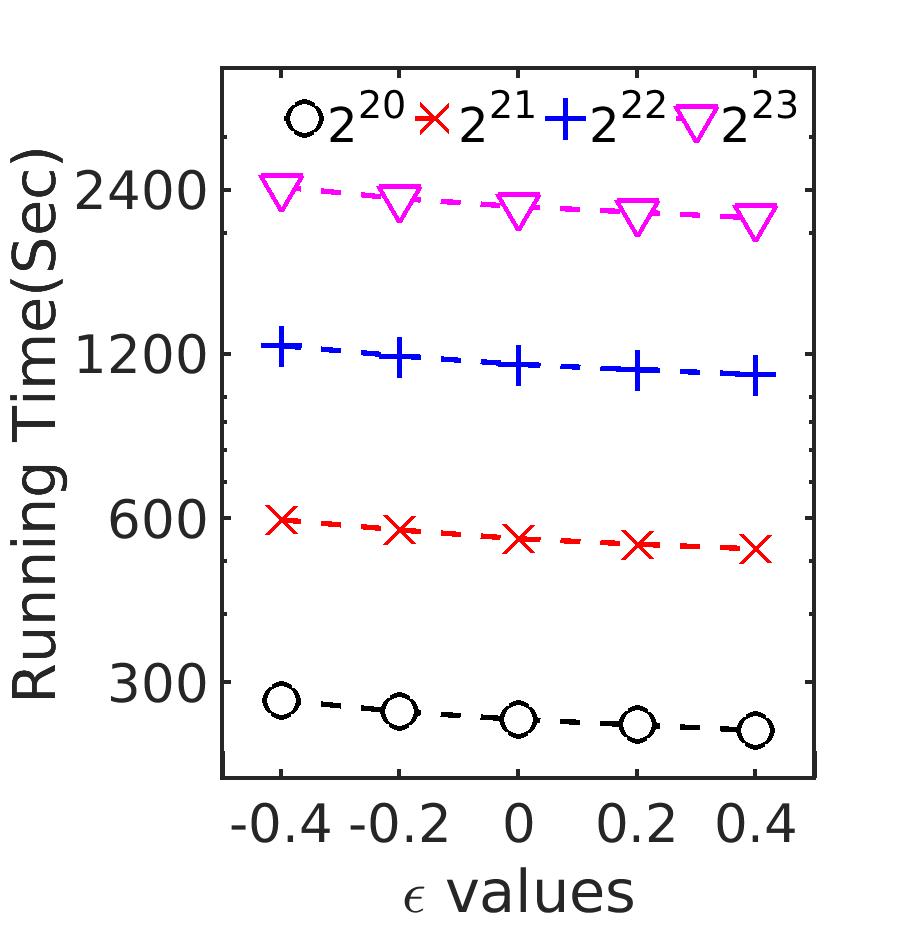}\label{fig:frdrunrmat}
\label{fig:frdruntimermat}
}\hspace{-15pt}
\subfigure[\scriptsize Random networks]{
\includegraphics[height=1.1in]{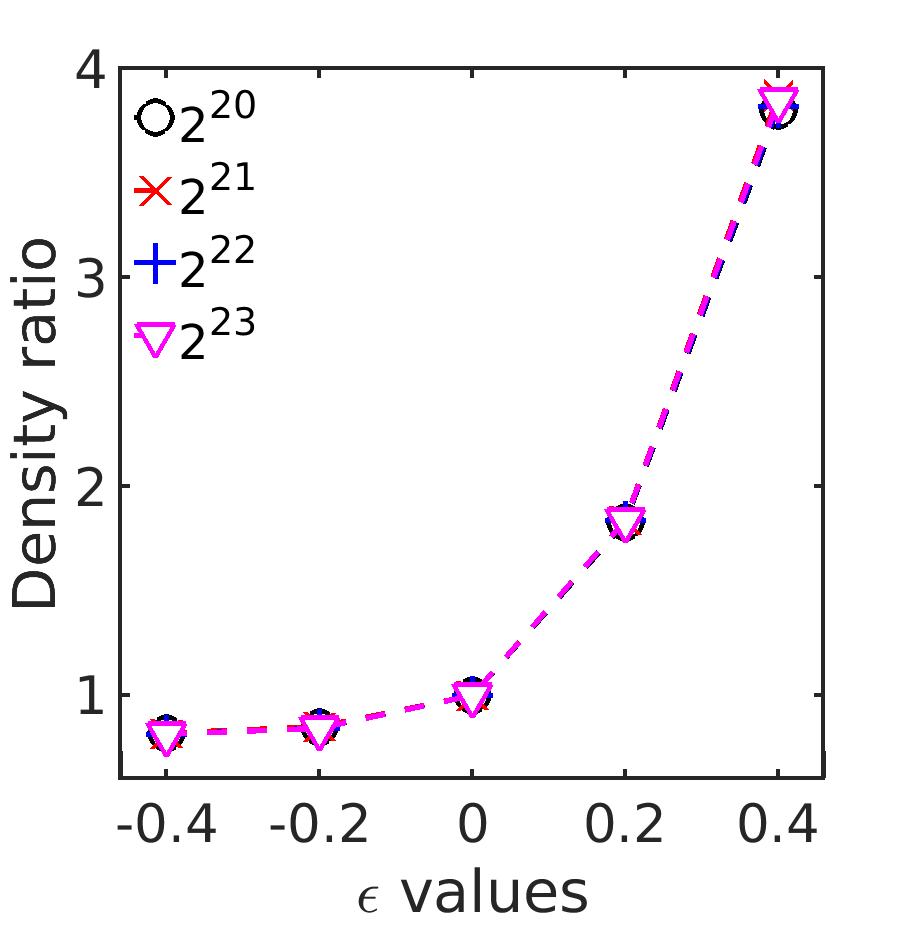}
\label{fig:frdapproxrandom}
}\hspace{-15pt}
\subfigure[\scriptsize R-MAT networks]{
\includegraphics[height=1.1in]{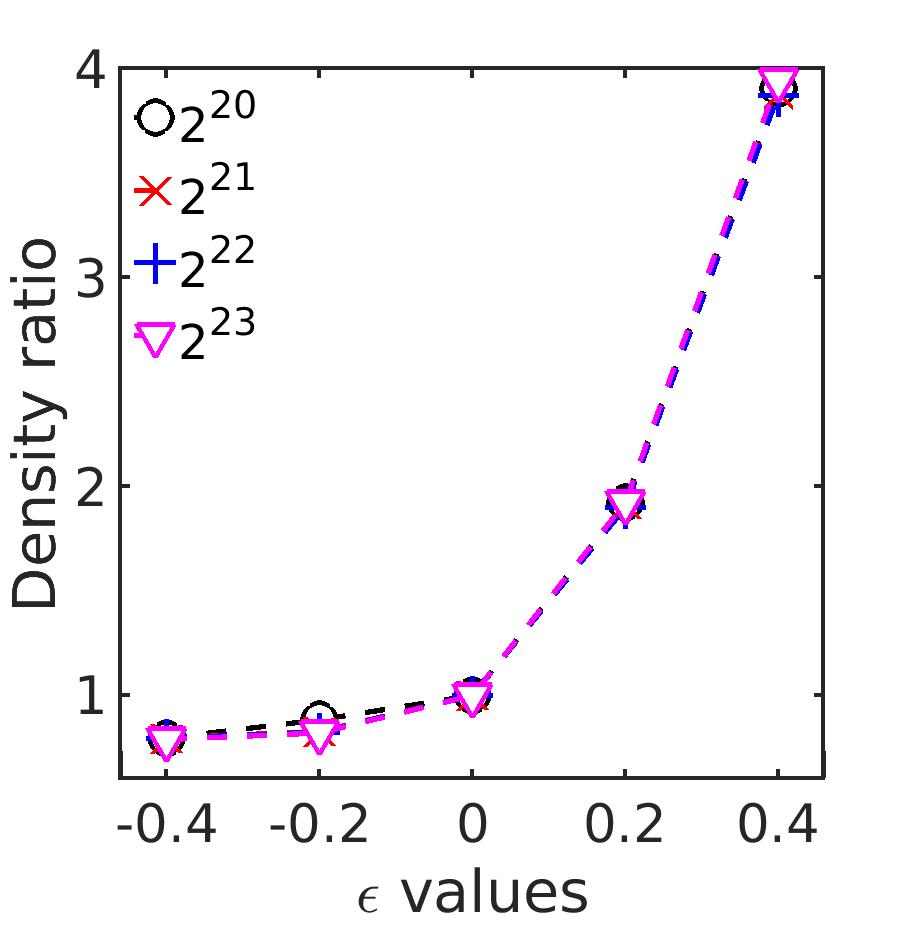}
\label{fig:frdapproxrmat}
}
\vspace{-15pt}
\caption{FRD evaluations for $\epsilon \in [-0.4,0.4]$  }\label{fig:frdruntimes}
\vspace{-2pt}
\end{figure}

Figures \ref{fig:frdruntimerandom} and \ref{fig:frdruntimermat} represent the running-times of FRD. The x axis represents different $\epsilon$ values and the y axis represents running-times in log scale of seconds. Each point represents running-time of FRD for given network and $\epsilon$ configurations. Increase in $\epsilon$ value causes higher amount of deletion per pass, resulting in fewer passes. Hence, the running-times decrease with the increase of $\epsilon$.

Figures \ref{fig:frdapproxrandom} and \ref{fig:frdapproxrmat} represent the density change of resultant CDC subgraphs for given $\epsilon$ value, with respect to $\epsilon=0$. The x axis represents different $\epsilon$ values, and the y axis represents the ratio, Density of CDC for $\epsilon=0$/Density of CDC with given $\epsilon$. Each point represents this density ratio obtained by FRD, for given network and $\epsilon$ configurations. Higher $\epsilon$ values result in more deletions per pass, lowering the densities of the resulting CDC subgraphs. Hence, the density ratio increases as the $\epsilon$ value decreases. We observe similar trends from real networks. The densities of resultant CDC subgraphs obtained by FRD depend on network topologies. Hence, for the same type of synthetic networks with the same $\epsilon$ value, the variance in the density ratio is low. 

\section{Conclusion}
In this paper, we introduce Triple Network, its CDC subgraph problem and its variants. We provide heuristics to find feasible solutions to these patterns, otherwise NP-Hard to find. We conclude that CDC subgraphs yield communities with similar charasteristics by illustrating the information gain of these patterns in NYC taxi, Twitter, ArnetMiner, and Flixter networks. We demonstrate the efficiency of our algorithms on large real and synthetic networks by observing running-time and density trends in real and synthetic networks.

\bibliographystyle{splncs04}

\bibliography{09References.bib}

\balance

%
%
%
\end{document}